\documentclass[11pt,onecolumn]{article}
\usepackage{amssymb}
\usepackage{amsmath}
\usepackage{setspace}
\onehalfspace

\usepackage{latexsym}
\usepackage{graphicx}
\usepackage{epsf}
\usepackage{epstopdf}

\normalsize \makeatletter
\newenvironment{proof}{{\em Proof.} }{\hfill$\Box$\vspace{0.1in}}

\newtheorem{lemma}{Lemma}[section]

\newtheorem{theorem}{Theorem}[section]

\newtheorem{thm}{Theorem}[section]

\newtheorem{lem}[thm]{Lemma}

\newcommand{\head}[1]
{\markright{\hbox to 0pt{\vtop to 0pt{\hbox{}\vskip 3mm \hrule width
\textwidth \vss} \hss}{\sc #1}}}

\begin{document}
\title{On the NP-hardness of scheduling with time restrictions}
\author{An
Zhang\thanks{Department of Mathematics,
Hangzhou Dianzi University, Hangzhou 310018, P. R. China. Supported
by the Zhejiang Provincial Natural Science Foundation of China (LY16A010015)} \and Yong Chen\thanks{Department of Mathematics, Hangzhou
Dianzi University, Hangzhou 310018, P. R. China. Supported by the
National Natural Science Foundation of China (11401149).} \and Lin Chen\thanks{Institute for Computer Science and Control, Hungarian Academy of Science (MTA SZTAKI), Budapest, Hungary.} \and Guangting Chen\thanks{Corresponding author. Taizhou University,
Linhai 317000, Zhejiang, China.  Supported
by the National Natural Science Foundation of China (11571252). {\tt gtchen@hdu.edu.cn} } }
\date{}
\maketitle \baselineskip 19pt

\begin{abstract}
In a recent paper, Braun, Chung and Graham \cite{BCG14} have addressed a single-processor scheduling problem with time
restrictions. Given a fixed integer $B\geq 2$, there is a set of jobs to be processed by a single processor subject
to the following $B$-constraint. For any real $x$, no unit
time interval $[x, x+1)$ is allowed to intersect more than $B$ jobs.
The problem has been shown to be NP-hard when $B$ is part of the input and left as an open question whether it remains NP-hard or not if $B$ is fixed \cite{BCG14,BCG16,ZYCC16}. This paper contributes to answering this question that we prove the problem is NP-hard even when $B=2$. A PTAS is also presented for any constant $B\geq 2$.

\vskip 2mm\noindent\textbf{Keywords:} Single-processor scheduling;
Time restrictions; NP-hardness \vskip
2mm\noindent\textbf{Mathematics Subject Classification(2000):}
90B35, 90C27
\end{abstract}

\section{Introduction}
Recently, Braun, Chung and Graham \cite{BCG14} have addressed a new single processor problem, namely \textit{scheduling with time restrictions}. Given a set of jobs and a single processor, the problem is therefore to schedule jobs sequentially on the processor so that the makespan of the schedule is minimized and the following $B$-constraint is satisfied. For any real $x$, no unit time interval $[x, x+1)$ is allowed to intersect more than $B$ jobs, where $B\geq 2$ is a given integer. Generally, it takes a semi-open time interval of length $s_i$ on the processor to finish a job with execution time $s_i$ and the processor can handle at most one job at a time. However, it is of interest to allow jobs with zero execution time in this problem, since some of these \textit{zero-jobs} could be processed at the same point in time, which makes the $B$-constraint more challenging.

Different from those existed models in the literature \cite{L04,BLK83,EOO13}, one can comprehend the $B$-constraint as a new type of resource constraint in scheduling problems. Suppose there are $B$ units of discrete and renewable resources, and each job requires not only the processor but also one unit of such resources for processing. The resource will be used up once the job is completed and it can be renewed in one unit time. With these settings, no unit time interval can admit more than $B$ jobs. Therefore, it turns out to be a new type of resource constraint. In \cite{ZYCC16}, Zhang et al further observed a scenario where the problem of scheduling with time restrictions might happen in Operating Room Scheduling (ORS) system of a hospital.

Benmansour, Braun and Artiba \cite{BBA14} formulated the single processor problem to a MILP model and tested the performance by running it on randomly
generated instances. Braun, Chung and Graham \cite{BCG14} showed that any reasonable schedule $T$ must have $C^T\leq (2-\frac{1}{B-1})C^*+3$ if $B\geq 3$ and
$C^T\leq \frac{4}{3}C^*+1$ if $B=2$, where $C^T$ and $C^*$ denote the makespan of $T$ and the optimal schedule respectively. In \cite{BCG16}, it is further improved to $C^T\leq (2-\frac{1}{B-1})C^*+\frac{B}{B-1}$ for any $B\geq 3$. If we produce the schedule $T$ by the LPT algorithm (in which jobs are arranged by non-increasing order of their execution times), then it is shown that $C^T\leq (2-\frac{2}{B})C^*+1$ for any $B\geq 2$ and the bound is best possible \cite{BCG16}. In \cite{ZYCC16}, improved algorithms are presented so that the generated schedule $T$ always satisfies that $C^T\leq C^*+\frac{1}{2}$ if $B=2$, $C^T\leq \frac{B}{B-1}C^*+2$ if $B=3, 4$ and $C^T\leq \frac{5}{4}C^*+2$ if $B\geq 5$. It is worth noting that though many algorithmic results have been acquired, little is known on the complexity of this problem. By now, it is only known that the single processor problem is NP-hard if $B$ is part of the input \cite{BCG14}. It has been left as an open question in \cite{BCG14,BCG16,ZYCC16} whether the problem is NP-hard or not if $B\geq 2$ is a fixed integer.

In this paper, we solve the open question by showing that the single processor problem is NP-hard even when $B=2$. Also we provide a PTAS for any constant $B\geq 2$. The rest of the paper is organized as follows. In Section 2, we formally describe the problem and give several useful lemmas. Section 3 is dedicated to the NP-hardness proof. In Section 4, we give the PTAS. Finally, some conclusion is made in Section 4.

\section{Problem description}
We are given a set $S=\{S_1, S_2, \cdots, S_n\}$ of jobs to be scheduled on a single processor. The execution time of job $S_i$ is denoted by $s_i\geq 0$, which has to be worked on during a semi-open time interval $[\alpha, \alpha+s_i)$ for some
$\alpha\geq 0$. Other than the zero-jobs, at most one job can be handled by the processor at any time point. Also, the $B$-constraint must be guaranteed during the process. The goal is to minimize the makespan, i.e., the latest completion time of the jobs. Without loss of generality, assume that $s_i\leq 1$ for each $i=1, 2, \cdots, n$.

Due to the $B$-constraint, a reasonable schedule on the single processor could be produced in the following way \cite{BCG14}.
Firstly, choose a permutation of $S$, say $T=(S_{[1]}, S_{[2]}, \cdots,
S_{[n]})$, then place the jobs one by one on the real line as early as
possible, provided that the $B$-constraint can never be violated.
More precisely, we start the first job $S_{[1]}$ at time $0$ and observe that
if some job $S_{[i]}$ finishes at time $\tau$ and
$\sum_{j=1}^{B-1}s_{[i+j]}\leq 1$, then the job $S_{[i+B]}$ cannot start
until time $\tau+1$. Let $C_{[i]}$ be the completion time of job $S_{[i]}$, then it follows that
\begin{equation}\label{eq1}
C_{[i]}=\left\{
  \begin{array}{ll}
    \sum_{j=1}^is_{[j]}, & \hbox{if $1\leq i\leq B$;} \\
    \max\{C_{[i-B]}+1, C_{[i-1]}\}+s_{[i]}, & \hbox{if $B<i\leq n$.}
  \end{array}
\right.
\end{equation}
We assume that there are at least $B+1$ jobs, whereas $C^T\geq 1$ for any permutation $T$. Let $d(i:j)$ be the length of the time interval between the completion time of job
$S_{[i]}$ and the start time of $S_{[j]}$ (see Figure 1), where $j>i$. Note that $d(i:i+1)$ equals exactly the length of idle times between $S_{[i]}$ and $S_{[i+1]}$. Let's simply denote by $d(i)=d(i:i+1)$, clearly we have the following observation.
\begin{lem}\label{prop0}
A permutation $T=(S_{[1]}, S_{[2]}, \cdots,
S_{[n]})$ is feasible if and only if for any $1\leq i\leq n-B$, $d(i:i+B)=d(i)+\sum_{j=i+1}^{i+B-1}(s_{[j]}+d(j))\geq 1$.
\end{lem}

\begin{center}
\begin{figure*}[h]
\centering
\includegraphics[width=1.0\textwidth]{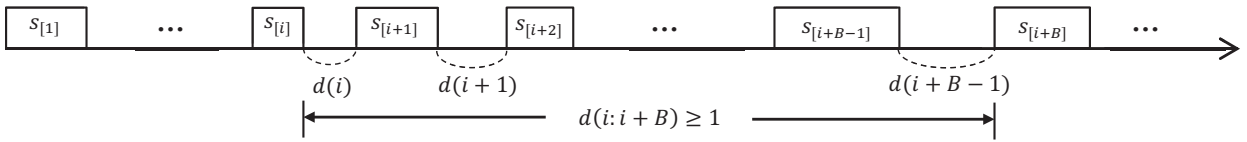}
\caption{The definition of $d(i:j)$.}
\end{figure*}
\end{center}

By the interchange method between jobs and the symmetrical characteristic of any permutation, we can get another observation.
\begin{lem}\label{prop1}
There exists an optimal permutation in which the first and the last jobs are exactly the two jobs with the first and the second smallest execution times.
\end{lem}

Here we introduce a technical result that will be used later.

\begin{lem}\label{techprop}
Given two sequences of $m$ real numbers, $\mathcal{X}=(x_1, x_2, \cdots, x_m)$ and $\mathcal{Y}=(y_1, y_2, \cdots, y_m)$,  such that:

(1) $x_1\leq x_2 \leq \cdots \leq x_m$ and $y_1\geq y_2 \geq \cdots \geq y_m$.

(2) $\sum_{j=1}^mx_j\leq \sum_{j=1}^m y_j$.

\noindent Then for any $1\leq i\leq m$, $\sum_{j=1}^ix_j\leq \sum_{j=1}^i y_j$.
\end{lem}
\begin{proof}
By (1), (2) and the fact that the average of the $i$ smallest (largest) number cannot be larger (smaller) than the average of all numbers, we can deduce that
$$
\frac{1}{i}\sum_{j=1}^ix_j\leq \frac{1}{m}\sum_{j=1}^mx_j\leq \frac{1}{m}\sum_{j=1}^m y_j\leq \frac{1}{i}\sum_{j=1}^i y_j.
$$
Hence we have $\sum_{j=1}^ix_j\leq \sum_{j=1}^i y_j$,  for any $1\leq i\leq m$.
\end{proof}


\section{NP-hardness of the single-processor problem}
For convenience, let us enlarge any unit time interval to a time interval with length $U\geq 1$, where $U$ is a given integer. Consequently, an equivalent definition of the $B$-constraints could be the following. For any real $x$, no
time interval $[x, x+U)$ is allowed to intersect more than $B$ jobs. The main result of this section is as follows.
\begin{thm}\label{main}
The single-processor problem with time restrictions is NP-hard even when $B=2$.
\end{thm}

We prove it by reducing the cardinality constrained partition problem, which is known to be NP-complete \cite{GTW88}, in a polynomial time into our problem.

\textbf{Cardinality Constrained Partition Problem:}

Given a set $\mathcal{N}=\{a_1, a_2, \cdots, a_{2m}\}$ of positive integers, does there exist a partition of $\mathcal{N}$ into two disjoint subsets $\mathcal{N}_1$ and $\mathcal{N}_2$ such that $\sum_{a_j\in \mathcal{N}_1}a_j=\sum_{a_j\in \mathcal{N}_2}a_j$ and that $|\mathcal{N}_1|=|\mathcal{N}_2|=m$?

The corresponding single-processor problem can be constructed as follows.

\noindent Basic settings: $B=2$; $U=\frac{1}{2}\sum_{a_j\in \mathcal{N}}a_j$.

\noindent Number of jobs: $n=2m+3$.

\noindent Execution times: $s_j=a_j, j=1, 2, \cdots, 2m$; $s_{2m+1}=s_{2m+2}=0$; $s_{2m+3}=U$.

\noindent Threshold of the makespan: $y=(m+2)U$.

\noindent Question: does there exist a permutation of the single-processor problem such that the makespan is no more than $y$?

\begin{lem}\label{ifpart}
If there exists a solution to the cardinality constrained partition problem, then there exists a permutation for the single-processor problem with makespan no more than $y$.
\end{lem}
\begin{proof}
Let $\mathcal{N}_1$ and $\mathcal{N}_2$ be the solution to the partition problem, we simply denote them by $\mathcal{N}_1=\{a_{[2j-1]}| j=1, 2, \cdots, m\}$ and  $\mathcal{N}_2=\{a_{[2j]}| j=1, 2, \cdots, m\}$, then $\sum_{j=1}^m a_{[2j-1]}=\sum_{j=1}^m a_{[2j]}=U$. Moreover, let $a_{[1]}\geq a_{[3]}\geq\cdots\geq a_{[2m-1]}$ and $a_{[2]}\leq a_{[4]}\leq\cdots\leq a_{[2m]}$. By Lemma \ref{techprop}, we get
\begin{equation}\label{eq2}\sum_{j=1}^i a_{[2j]} \leq \sum_{j=1}^i a_{[2j-1]}\end{equation}
for any $1\leq i\leq m$. Consider the following
permutation (See Figure 2),
$$T=(S_{2m+1}, S_{[1]}, S_{[2]}, \cdots, S_{[2m-1]}, S_{[2m]}, S_{2m+3}, S_{2m+2}),$$
we claim that $C^T\leq y$.
\begin{center}
\begin{figure*}[h]
\centering
\includegraphics[width=1.0\textwidth]{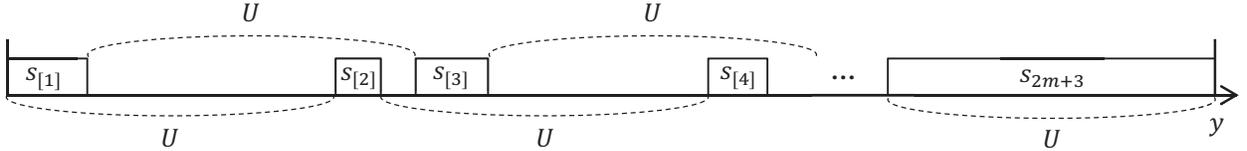}
\caption{Illustration for the schedule generated by $T$.}
\end{figure*}
\end{center}

In fact, by (\ref{eq1}), we have $C_{2m+1}=s_{2m+1}=0, C_{[1]}=s_{2m+1}+s_{[1]}=a_{[1]}$ and subsequently, $C_{[2]}=\max\{C_{2m+1}+U, C_{[1]}\}+s_{[2]}=\max\{U, a_{[1]}\}+a_{[2]}=U+a_{[2]}$. Assume that we already have $C_{[2i-1]}=(i-1)U+\sum_{j=1}^{i}a_{[2j-1]}$ and $C_{[2i]}=iU+\sum_{j=1}^{i}a_{[2j]}$. Then by (\ref{eq1}) and (\ref{eq2}), we can calculate
\begin{eqnarray*}
C_{[2i+1]} & = & \max\{C_{[2i-1]}+U, C_{[2i]}\}+s_{[2i+1]}  \\
    & = & \max\{iU+\sum_{j=1}^{i}a_{[2j-1]}, iU+\sum_{j=1}^{i}a_{[2j]}\}+a_{[2i+1]} \\
    & = & iU+\sum_{j=1}^{i}a_{[2j-1]}+a_{[2i+1]}=iU+\sum_{j=1}^{i+1}a_{[2j-1]},
\end{eqnarray*}
and
\begin{eqnarray*}
C_{[2i+2]} & = & \max\{C_{[2i]}+U, C_{[2i+1]}\}+s_{[2i+2]}  \\
    & = & \max\{(i+1)U+\sum_{j=1}^{i}a_{[2j]}, iU+\sum_{j=1}^{i+1}a_{[2j-1]}\}+a_{[2i+2]} \\
    & = & (i+1)U+\sum_{j=1}^{i+1}a_{[2j]}.
\end{eqnarray*}
Thus it yields that $C_{[2m-1]}=(m-1)U+\sum_{j=1}^{m}a_{[2j-1]}=mU$ and $C_{[2m]}=mU+\sum_{j=1}^{m}a_{[2j]}=(m+1)U$. Consequently, we get
$$C_{2m+3}=\max\{C_{[2m-1]}+U, C_{[2m]}\}+s_{2m+3}=(m+1)U+U=(m+2)U,$$
and hence, $C^T=C_{2m+2}=\max\{C_{[2m]}+U, C_{2m+3}\}+s_{2m+2}=(m+2)U=y$.
\end{proof}

\begin{lem}\label{onlyifpart}
If there exists a permutation for the single-processor problem with makespan no more than $y$, then there exists a solution to the cardinality constrained partition problem.
\end{lem}
\begin{proof}
Suppose $T$ is a permutation of $S$ such that $C^T\leq y$. By Lemma \ref{prop1}, we can assume that the two smallest jobs $S_{2m+1}$ and $S_{2m+2}$ are at the opposite ends of $T$. Therefore the permutation can be denoted by
$$T=(S_{2m+1}, S_{[1]}, S_{[2]}, \cdots, S_{[2m+1]}, S_{2m+2}),$$
where $\{S_{[1]}, S_{[2]}, \cdots, S_{[2m+1]}\}=\{S_{1}, S_2, \cdots, S_{2m}\}\cup\{S_{2m+3}\}$. By Lemma \ref{prop0}, we can deduce that
\begin{equation}\label{eq3}
C^T=C_{2m+2}\geq s_{2m+1}+U+s_{[2]}+U+s_{[4]}+\cdots+U+s_{[2m]}+U+s_{2m+2} = (m+1)U+\sum_{j=1}^{m}s_{[2j]}
\end{equation}
and
\begin{equation}\label{eq4}
C^T \geq C_{[2m+1]}+s_{2m+2}\geq s_{2m+1}+s_{[1]}+U+s_{[3]}+\cdots+U+s_{[2m+1]}+s_{2m+2} = mU+\sum_{j=1}^{m+1}s_{[2j-1]}.
\end{equation}
Combining (\ref{eq3}) and (\ref{eq4}), it follows
\begin{eqnarray*}
C^T & \geq & \max\{(m+1)U+\sum_{j=1}^{m}s_{[2j]}, mU+\sum_{j=1}^{m+1}s_{[2j-1]}\}  \\
& \geq & \frac{1}{2}\{(m+1)U+\sum_{j=1}^{m}s_{[2j]}+mU+\sum_{j=1}^{m+1}s_{[2j-1]}\}  \\
& = & \frac{1}{2}\{(2m+1)U+\sum_{j=1}^{2m}a_{j}+s_{2m+3}\}=(m+2)U=y.
\end{eqnarray*}
Since $C^T\leq y$, we must have $C^T=y$, which yields that (\ref{eq3}) and (\ref{eq4}) turn out to be equalities, i.e.,
$C^T=(m+1)U+\sum_{j=1}^{m}s_{[2j]}=mU+\sum_{j=1}^{m+1}s_{[2j-1]}.$ Hence, $\sum_{j=1}^{m+1}s_{[2j-1]}=U+\sum_{j=1}^{m}s_{[2j]}$. Note that $\{s_{[1]}, s_{[2]}, \cdots, s_{[2m+1]}\}=\{s_1, s_2, \cdots, s_{2m}, s_{2m+3}\}=\{a_1, a_2, \cdots, a_{2m}, U\}$ and $\sum_{j=1}^{2m}a_j=2U$, it yields that $U\in\{s_{[1]}, s_{[3]}, \cdots, s_{[2m+1]}\}$, say $s_{[2m+1]}=U$. Thus we have $\sum_{j=1}^ms_{[2j-1]}=\sum_{j=1}^ms_{[2j]}=U$. That is, $\mathcal{N}_1=\{a_{[2j-1]}|j=1, 2,\cdots, m\}, \mathcal{N}_2=\{a_{[2j]}|j=1, 2,\cdots, m\}$ forms a solution to the cardinality constrained partition problem.
\end{proof}

Finally, by Lemma \ref{ifpart} and Lemma \ref{onlyifpart}, Theorem 3.1 follows.

\section{A PTAS for any constant $B$}

The goal of this section is to prove the following theorem.

\begin{theorem}\label{thm:ptas}
There exists a PTAS for the single-processor problem with time restrictions when $B$ is a constant.
\end{theorem}

Let $\epsilon>0$ be an arbitrarily small constant. Let $OPT$ be the makespan of the optimal solution for the given instance $I$, whereas $OPT\geq 1$. We modify instance $I$ in the following way. For any job $S_i$, if its execution time $s_i\leq \epsilon$, we round it up to $\epsilon$. If $\epsilon<s_i\le 1$, we round it up to the nearest value of the form $\epsilon(1+\epsilon)^k$ where $k=1,2,\cdots, \lceil\frac{\log(1/\epsilon)}{\log(1+\epsilon)}\rceil$. Denote by $S_i'$ the job with rounded execution time $s_i'$ and let $I'$ be the modified instance that consists of $S_i'$. Note that there are at most $\tau=O(1/\epsilon\cdot\log (1/\epsilon))=\tilde{O}(1/\epsilon)$ different execution times in $I'$.

We have the following observation.

\begin{lemma}\label{le:1}
There exists a feasible solution for $I'$ whose makespan is at most $OPT(1+O(\epsilon))$.
\end{lemma}
\begin{proof}
Consider the optimal solution $Sol$ for $I$, whose makespan is $OPT$. Suppose the permutation of jobs is $T=(S_{[1]},S_{[2]},\cdots,S_{[n]})$. We replace each job $S_{[i]}$ with $S_{[i]}'$ while keeping the order of the jobs as well as the idle time between each $S_{[i]}$ and $S_{[i+1]}$ intact. Let $Sol'$ be the modified solution. We claim that, the following two statements are true:
\begin{itemize}
\item[(a).] The modified solution $Sol'$ is a feasible solution for $I'$.
\item[(b).] The modified solution $Sol'$ has a makespan of $OPT(1+O(\epsilon))$.
\end{itemize}
If both statements are true, the lemma follows directly.
	
We prove the first statement (a). Suppose on the contrary that it is false, then there exists some interval $[x,x+1)$ that intersects more than $B$ jobs. By Lemma \ref{prop0}, there exists $1\leq i \leq n-B$ such that $d'(i:i+B)<1$, that is,
$$d'(i:i+B)=d(i)+\sum_{j=i+1}^{i+B-1}(s'_{[j]}+d(j))<1.$$
Since $s_{[j]}'\geq s_{[j]}$ for any $j$, it follows that
$$d(i:i+B)=d(i)+\sum_{j=i+1}^{i+B-1}(s_{[j]}+d(j))\leq d(i)+\sum_{j=i+1}^{i+B-1}(s'_{[j]}+d(j))=d'(i:i+B)<1,$$
implying that the solution $Sol$ is not feasible for $I$, which is a contradiction.
	
We prove the second statement (b). In \cite{ZYCC16}, the authors have shown that $OPT\geq \max\{\sum_{i=1}^n s_i, \frac{n-B}{B}\}$, thus it suffices to show that $\sum_{i=1}^n s_i'\leq \sum_{i=1}^n s_i+OPT\cdot O(\epsilon)$. Obviously we have
$$\sum_{i:s_i>\epsilon} s_i'\le (1+\epsilon)\sum_{i:s_i>\epsilon} s_i\le \sum_{i:s_i>\epsilon} s_i+\epsilon OPT.$$
Consider the jobs with $s_i\le \epsilon$, whose overall length is no more than $n\epsilon$. Clearly we have
$$\sum_{i:s_i\leq\epsilon} s_{i'}\leq n\epsilon\leq (OPT+1)B\cdot\epsilon\leq  OPT\cdot 2B\epsilon\leq OPT\cdot O(\epsilon). $$
\end{proof}

Let $Sol'$ be the feasible solution for $I'$ satisfying Lemma~\ref{le:1}. Notice that the idle time between any two adjacent jobs is at most $1$. We further modify $Sol'$ as follows. If the idle time $d(i)$ between $S'_{[i]}$ and $S'_{[i+1]}$ is no more than $\epsilon$, round it up to $\epsilon$. Otherwise round it up to the nearest value of the form $\epsilon(1+\epsilon)^k$ where $k=1,2,\cdots, \lceil\frac{\log(1/\epsilon)}{\log(1+\epsilon)}\rceil$. Clearly, there are also $\tau=\tilde{O}(1/\epsilon)$ different idle times. Let $Sol''$ be the modified solution, we have the following observation.

\begin{lemma}\label{le:2}
The makespan of $Sol''$ is $OPT(1+O(\epsilon))$.
\end{lemma}
The proof is essentially the same as that of Lemma~\ref{le:1}. Notice that no interval of $[x,x+1)$ can intersect more than $B+1$ idle times, thus we can view the idle time between two jobs as a job, and view the execution time of a job as an idle time. Applying the same argument as Lemma~\ref{le:1} suffices.

We call a solution for $I'$ as a regular solution if it satisfies that the idle time between any two jobs is of the form $\epsilon(1+\epsilon)^k$. Lemma~\ref{le:2} implies that the optimal regular solution has a makespan at most $OPT(1+O(\epsilon))$. In what follows, we provide a dynamic programming algorithm to find out the optimal regular solution for $I'$, whereas Theorem~\ref{thm:ptas} follows.

Denote by $V'$ the set of execution times and $U'$ the set of idle times between jobs. Recall that either $V'$ or $U'$ contains at most $\tau=\tilde{O}(1/\epsilon)$ different elements. Consider a $(\tau+2B)$-vector $z=(n_1,n_2,\cdots,n_{\tau},u_1,v_1,u_2,v_2\cdots,u_{B},v_B)$. We associate with the above vector the makespan of the optimal regular solution for the subproblem such that there are $n_i$ jobs of execution time $\epsilon(1+\epsilon)^{i-1}$, and the last $B$ jobs are scheduled as follows: it starts with an idle time of length $u_1$, followed by a job of execution time $v_1$, then an idle time of length $u_2$, followed by a job of execution time $v_2$, $\cdots$, the last job is of execution time $v_B$. A vector $z$ is called feasible if $u_1+v_1+u_2+v_2+\cdots+v_{B-1}+u_B\geq 1$, in which we mean that the generated schedule is feasible by Lemma \ref{prop0}.

Let $f(z)$ be the value associated with the vector $z$, then it could be calculated iteratively as follows.

\textbf{Initialization:} If $s_i\in\{v_1, v_2, \cdots, v_B\}$ but $n_i=0$ or $n_i<0$ for some $i$ or $\sum_{i=1}^\tau n_i>n$ or $z$ is infeasible, then $f(z)=+\infty$. Otherwise if $\sum_{i=1}^\tau n_i=B$, then $f(z)=\min_{u_i\in U'}\sum_{i=1}^B (u_i+v_i)$.

\textbf{Recursive function:} Suppose $z$ is feasible, $v_B=\epsilon(1+\epsilon)^{i-1}$ and $\sum_{i=1}^\tau n_i>B$. For any $1\leq x,y\leq \tau$, define $z_{xy} = (n_1,\cdots,n_{i-1},n_i-1,n_{i+1},\cdots,n_{\tau},\epsilon(1+\epsilon)^{x-1},\epsilon(1+\epsilon)^{y-1},u_1,v_1,\cdots,u_{B-1},v_{B-1})$.
$$
f(z)=\min_{1\leq x,y\leq \tau}\{f(z_{xy})+u_B+v_B\}.
$$

\textbf{Objective value:} $f^*=\min_{u_i\in U', v_i\in V'}\{f(z)|I'\ \textrm{consists of}\ n_i\ \textrm{jobs of length}\ \epsilon(1+\epsilon)^{i-1}\}$.

Since there are at most $n^{\tau}\cdot\tau^{2B}=n^{\tilde{O}(1/\epsilon)}\cdot \tilde O(1/\epsilon^{2B})$ different kinds of vectors, the time complexity of the above dynamical programming could be bounded by $n^{\tilde{O}(1/\epsilon)}\cdot \tilde O(1/\epsilon^{2B})$.

\section{Conclusion}
We studied the single-processor scheduling problem with time restrictions. In previous work \cite{BCG16,ZYCC16}, asymptotically optimal permutations have been acquired for $B=2$. Now this paper sends a proof that the problem becomes NP-hard even when $B=2$. It makes the problem fairly interesting. Though there exists a PTAS for any constant $B$, it remains open whether the problem is strongly NP-hard or not for a constant $B\geq 3$ or for an input $B$.

\end{document}